\documentclass[conference]{IEEEtran}
\usepackage{graphicx}
\usepackage{amsthm}
\usepackage{epsfig}
\usepackage{latexsym}
\usepackage{amsfonts}
\usepackage{here}
\usepackage{rawfonts}
\usepackage[latin1]{inputenc}
\usepackage[T1]{fontenc}
\usepackage{calc}
\usepackage{capitalgreekitalic}
\usepackage{url}
\usepackage{enumerate}
\usepackage{color}
\usepackage[tbtags]{amsmath}
\usepackage{amssymb}
\usepackage{upref}
\usepackage{epic,eepic}
\usepackage{times}
\usepackage{dsfont}
\usepackage{comment}
\usepackage{cite}
\usepackage{bbm} 

\renewcommand{\Pr}{\ensuremath{\operatorname{Pr}}}

\newtheorem{theorem}{\bf Theorem}

\usepackage{dsfont}

\newcounter{step}
\newlength{\totlinewidth}
  {\end{list}%
  \rule{\linewidth}{1pt}}
\newcounter{substep}

  {\end{list}}

\newlength{\aligntop}
\newlength{\alignbot}
 

\IEEEoverridecommandlockouts

\usepackage{algorithm}
\usepackage{algorithmic}

\usepackage{subfigure}

\begin{document}
\clearpage
\title{\Huge Echo State Transfer Learning for Data Correlation Aware Resource Allocation in Wireless Virtual Reality \vspace*{-0.5em}}
\author{{Mingzhe Chen\IEEEauthorrefmark{1}}, Walid Saad\IEEEauthorrefmark{2}, Changchuan Yin\IEEEauthorrefmark{1}, and M\'erouane Debbah\IEEEauthorrefmark{3}\vspace*{0em}\\ 
\authorblockA{\small \IEEEauthorrefmark{1}Beijing Key Laboratory of Network System Architecture and Convergence,\\ Beijing University of Posts and Telecommunications, Beijing, China 100876, Emails: \protect\url{chenmingzhe@bupt.edu.cn}, \protect\url{ccyin@ieee.org.} \\
\IEEEauthorrefmark{2}Wireless@VT, Bradley Department of Electrical and Computer Engineering, Virginia Tech, Blacksburg, VA, USA, Email: \protect\url{walids@vt.edu.}\\
\IEEEauthorrefmark{3}\small Mathematical and Algorithmic Sciences Lab, Huawei France R \& D, Paris, France, \\Email: merouane.debbah@huawei.com.\\
}\vspace*{-3.5em}
\thanks{This work was supported in part by the National Natural Science Foundation of China under Grants 61671086 and 6162910, by the U.S. National Science Foundation under Grants IIS-1633363, CNS-1460316, and CNS-1617896 and by the ERC Starting 
Grant 305123 MORE (Advanced Mathematical Tools for Complex Network Engineering).}
 }

\maketitle

\vspace{0cm}
\begin{abstract}
In this paper, the problem of data correlation-aware resource management is studied for a network of wireless virtual reality (VR) users communicating over cloud-based small cell networks (SCNs). In the studied model, small base stations (SBSs) with limited computational resources act as VR control centers that collect the tracking information from VR users over the cellular uplink and send them to the VR users over the downlink. In such a setting, VR users may send or request correlated or similar data (panoramic images and tracking data). 
This potential spatial data correlation can be factored into the resource allocation problem to reduce the traffic load in both uplink and downlink. This VR resource allocation problem is formulated as a noncooperative game that allows jointly optimizing the computational and spectrum resources, while being cognizant of the data correlation. To solve this game, a transfer learning algorithm based on the machine learning framework of echo state networks (ESNs) is proposed. Unlike conventional reinforcement learning algorithms that must be executed each time the environment changes, the proposed algorithm can intelligently transfer information on the learned utility, across time, to rapidly adapt to environmental dynamics due to factors such as changes in the users' content or data correlation. Simulation results show that the proposed algorithm  achieves up to 16.7\% and 18.2\% gains in terms of delay compared to the Q-learning with data correlation and Q-learning without data correlation. The results also show that the proposed algorithm has a faster convergence time than Q-learning and can guarantee low delays.\end{abstract}
%


\section{Introduction}
 Virtual reality (VR) can enable users to virtually hike the Grand Canyon or engage in a secret adventure as a video game hero without leaving their room. 
However, due to the wired connections of conventional VR devices, the users can only take a restricted set of actions which, in turn, limits the VR application space.
To enable immersive VR applications, VR systems can be operated using wireless networking technologies \cite{bacstuug2016towards}. However, operating VR devices over wireless small cell networks (SCNs) faces many challenges \cite{bacstuug2016towards} that include effective image compression, tracking, and low-latency computation and communication.

The existing literature has studied a number of problems related to wireless VR such as in \cite{bacstuug2016towards,ahn2017delay,singh2017high,VROWNchen}. The authors in \cite{bacstuug2016towards} exposed the future challenges of VR systems over a wireless network. However, this work is a qualitative survey that does not provide any rigorous wireless VR model. In \cite{ahn2017delay}, a channel access scheme for wireless multi-user VR system is proposed. The authors in \cite{singh2017high} proposed an alternate current magnetic field-based tracking system to track the position and orientation of a VR user's head. However, the recent works in \cite{ahn2017delay} and \cite{singh2017high} do not develop a VR model that can capture all factors of VR QoS and they only analyze a single VR metric such as delay or tracking accuracy. 
 In \cite{VROWNchen}, we proposed a wireless VR model that captures the tracking accuracy, processing delay, and transmission delay and proposed a machine learning based algorithm to solve the resource allocation problem. However, the work in \cite{VROWNchen} focuses only on spectrum allocation that ignores the correlation between the data of the VR users. Indeed, the VR data (tracking data or VR image data) pertaining to different users can be potentially correlated. For example, when the VR users are watching a football game with different perspectives, the cloud will need to only transmit one ${360^ \circ }$ image to the SBS, then the SBS can rotate the image and transmit it to different users. In this case, the use of data correlation to reduce the traffic load in the transmission of tracking information and VR images can improve the delay.

The main contribution of this paper is a novel framework for enabling VR applications over wireless cellular networks. To the best of our knowledge, \emph{this is the first work that jointly considers data correlation, spectrum resource allocation, and computational resource allocation for VR over cellular networks.} In this regard, our key contributions include:

\begin{itemize}
\item We propose a novel VR model to jointly capture the downlink and uplink transmission delay, backhaul transmission delay, and computational time thus effectively quantifying the VR delay for all users in a wireless VR network.
\item For the considered wireless VR applications, we analyze the allocation of resource blocks \emph{jointly} over, the uplink and downlink, along with the allocation of computational resources in the uplink. We formulate the problem as a noncooperative game in which the players are the small base stations (SBSs). Each player seeks to find an optimal resource allocation scheme to optimize a utility function that captures the VR delay. 
\item To find a Nash equilibrium of this game, we propose a transfer learning algorithm based on echo state networks (ESNs) \cite{chen2017machine}. The proposed algorithm can intelligently transfer information on the learned utility across time, and, hence, allow adaptation to environmental dynamics due to factors such as changes in the users' data correlation.
\item Simulation results show that the proposed algorithm can, respectively, yield 16.7\% and 18.2\% gains in terms of delay compared to Q-learning with data correlation and Q-learning without data correlation.   
\end{itemize}
 

\section{System Model and Problem Formulation}\label{se:system}
%

Consider the downlink and uplink transmissions of an SCN servicing a set $\mathcal{U}$ of $U$ wireless VR users and a set $\mathcal{B}$ of $B$ SBSs. Here, the downlink is used to transmit the VR images displayed on each user's VR device while the uplink is used to transmit the tracking information that is used to determine each VR user's location and orientation. The SBSs are connected to a cloud via \emph{capacity-constrained backhaul} links and the SBSs serve their users using the cellular band. $V_F$ represents the maximum backhaul transmission rate for all users. Here, we focus on entertainment VR applications such as watching immersive videos and playing immersive games.

In our model, the SBSs adopt an orthogonal frequency division multiple access (OFDMA) technique and transmit over a set of $\mathcal{V}$ of $V$ uplink resource blocks and a set of $\mathcal{S}$ of $S$ downlink resource blocks. The coverage of each SBS is a circular area with radius $r$ and each SBS only allocates resource blocks to the users located in its coverage range. We also assume that the resource blocks of each SBS will all be allocated to the associated users. 

 
   
\subsection{Data Correlation Model}
\subsubsection{Downlink Data Correlation Model}          
In VR wireless networks, multiple VR users may play the same immersive game with different locations and orientations. In this case, the cloud can exploit the data correlation between the users that are playing the same immersive game to reduce the traffic load of backhaul links. For example, when the users are watching the same immersive sports game, the cloud can extract the difference between the VR images of these users and will need to only transmit to an SBS the data that is unique to each user. However, when the VR users are playing different immersive games, the data correlation between the users is low and, hence, the cloud needs to transmit entire VR images to the associated VR users. In order to define the data correlation of VR images, we first assume that the number of pixels that user $i$ needs to construct the VR images is $N_{i}$ and the number of different pixels between any pair of users $i$ and $k$ is $N_{ik}$. Here, $N_{ik}$ is calculated by the cloud using image processing methods such as motion search \cite{yang2002computation}. Then, the data correlation between user $i$ and user $k$ can be defined as follows:
\begin{equation}\label{eq:datacu}
\setlength{\abovedisplayskip}{3 pt}
\setlength{\belowdisplayskip}{3 pt}
{\phi _{ik}} = \frac{{{N_{ik}}}}{{{N_i+N_k}}},
\end{equation}
where $N_k$ is the number of pixels that user $k$ needs to construct the VR images during a period. Indeed, (\ref{eq:datacu}) captures the difference between the images of users $i$ and $j$. From (\ref{eq:datacu}), we can see that, when user $i$ and user $k$ are associated with the same SBS, the cloud needs to only transmit $N_i+N_j-\left(N_i+N_j\right){\phi _{ij}}$ pixels to that SBS.
\subsubsection{Uplink Data Correlation Model}   
 In the uplink, the users must transmit the tracking information to the SBSs. The tracking information is collected by the sensors placed at a VR user's headset or near the VR user. It has been shown that, for commonly used data-gathering applications, the data source can be modeled as a Gaussian field \cite{cressie2015statistics}. The uplink data is collected by the sensors and, hence, the uplink data can be assumed to follow the Gaussian distribution. We can assume that the tracking data, $X_i$, collected by each VR user $i$ is a Gaussian random variable with mean $\mu_i$ and variance $\sigma_i^2$. In wireless VR, observations from proximal VR devices are often correlated due to the dense deployment density. Hence, we consider the power exponential model \cite{vuran2004spatio} to capture the spatial correlation of VR tracking data. Here, the covariance $\sigma_{ij}$ between user $i$ and user $j$ separated by distance $d_{ij}$ is: 
 \begin{equation}
 \setlength{\abovedisplayskip}{4 pt}
\setlength{\belowdisplayskip}{4 pt}
 {\sigma _{ij}} = {\mathop{\rm cov}} \left( {{X_i},{X_j}} \right) = {\sigma _i}{\sigma _j}{e^{ - {{d_{ij}^\alpha } \mathord{\left/
 {\vphantom {{d_{ij}^\alpha } \kappa }} \right.
 \kern-\nulldelimiterspace} \kappa }}},
\end{equation}    
 where $\alpha$ and $\kappa$ capture the significance of distance variation on data correlation. 
 
 \subsection{Delay Model}
In our model, the VR images are transmitted from the cloud to the SBSs then to the users. The tracking information is transmitted from the users to the SBSs and processed at each corresponding SBS. In this case, the backhaul links are only used for VR image transmission and the transmission rate of each VR image from the cloud to the SBS can be given as ${V_{Fi}} = \frac{{{V_{F}}}}{U}$ \cite{chen2016caching}. Here, we assume that the backhaul transmission rate of each user is equal and we do not consider the optimization of the backhaul transmission. In a VR model, we need to capture the VR transmission requirements such as high data rate, low delay, and accurate tracking and, hence, we consider the transmission delay as the main VR QoS metric of interest. The downlink rate of user $i$ associated with SBS $j$ is:
\begin{equation}\label{eq:cd}
\setlength{\abovedisplayskip}{4 pt}
\setlength{\belowdisplayskip}{4 pt}
{c_{ij}\left(\boldsymbol{s}_{ij}\right)} = \sum\limits_{k = 1}^{{S}} {s_{ij,k}B{{\log }_2}\left( {1 + \gamma _{ij,k}} \right)},
\end{equation}
where $\boldsymbol{s}_{ij} = \left[ {s_{ij,1}, \ldots ,s_{ij,{S}}} \right]$ is the vector of resource blocks that SBS $j$ allocates to user $i$ with $s_{ij,k} \in \left\{ {1,0} \right\}$. Here, $s_{ij,k}=1$ indicates that resource block $k$ is allocated to user $i$. ${\gamma _{ij,k}=\frac{{{P_B}{h_{ij}^k}}}{{{N_0 ^2} + \sum\limits_{l \in \mathcal{R}^k,l \ne j} {{P_B}{h_{il}^k}} }}}$ is the signal-to-interference-plus-noise ratio (SINR) between user $i$ and SBS $j$ over resource block $k$. $\mathcal{R}^k$ represents the set of the SBSs that use downlink resource block $k$, $B$ is the bandwidth of each subcarrier, $P_B$ is the transmit power of SBS $j$ which is assumed to be equal for all SBSs, $N_0^2$ is the variance of the Gaussian noise and $h_{ij}^k=g_{ij}^kp_{ij}^{-\beta}$ is the path loss between user $i$ and SBS $j$ over resource block with $g_{ij}^k$ is the Rayleigh fading parameter, $d_{ij}$ is the distance between user $i$ and SBS $j$, and $\beta$ is the path loss exponent. Based on (\ref{eq:datacu}) and (\ref{eq:cd}), the downlink transmission delay at time slot $t$ is:
\begin{equation}
\setlength{\abovedisplayskip}{4 pt}
\setlength{\belowdisplayskip}{4 pt}
 D_{ij}\left(L_{i}\left(\phi _{i}^{\max} \right),\boldsymbol{s}_{ij}\right) = \frac{{{L_{i}\left(\phi _{i}^{\max} \right)}}}{{{c_{ij}\left(\boldsymbol{s}_{ij}\right)}}}+\frac{{{L_{i}\left(\phi _{i}^{\max} \right)}}}{V_{Fi}},
 \end{equation}
  where $L_{i}\left(\phi _{i}^{\max} \right)$ is the data that user $i$ needs to construct a VR image during a period and $\phi _{i}^{\max}\!\!=\!\!\!\mathop {\max }\limits_{k\in \mathcal{U}_j, k \ne i}\! \left( {{\phi _{ik}}} \right)$ is the maximum downlink data correlation between user $i$ and other users associated with SBS $j$. Finding the maximum data correlation allows minimizing the downlink transmission data transmitted in the downlink and that will be used construct a VR image. Here, the first term is the transmission time from SBS $j$ to user $i$ and the second term is the transmission time from the cloud to SBS $j$. We assume that $P_U$ is the transmit power of each user which is assumed to be equal for all users. The bandwidth of each uplink resource block is also $B$. In this case, the uplink rate of each user $i$ associated with SBS $j$ is:
\begin{equation}
\setlength{\abovedisplayskip}{4 pt}
\setlength{\belowdisplayskip}{4 pt}
{c_{ij}\left(\boldsymbol{v}_{ij}\right)} = \sum\limits_{k = 1}^{{V}} {v_{ij,k}B{{\log }_2}\left( {1 + \gamma _{ij,k}^\textrm{u}} \right)},
\end{equation}
where $\boldsymbol{v}_{ij} = \left[ {v_{ij,1}, \ldots, v_{ij,{V}}} \right]$ is the vector of resource blocks that SBS $j$ allocates to user $i$ with $v_{ij,k} \in \left\{ {1,0} \right\}$. ${\gamma _{ij,k}^\textrm{u}=\frac{{{P_U}{h_{ij}^k}}}{{{\sigma ^2} + \sum\limits_{l \in \mathcal{U}^k,l \ne j} {{P_U}{h_{il}^k}} }}}$ is the SINR between user $i$ and SBS $j$ over resource block $k$ with $\mathcal{U}^k$ represents the set of users that use uplink resource blocks $k$. In this case, the uplink transmission delay can be given by $\frac{{{K_{i}\left(\sigma_{i}^{\max}\right)}}}{{{c_{ij}\left(\boldsymbol{v}_{ij}\right)}}}$
where $K_{i}$ is the data {that needs to be transmitted} and $\sigma _{i}^{\max}=\mathop {\max }\limits_{k\in \mathcal{U}_j, k \ne i} \left( {{\sigma _{ik}}} \right)$ is the maximum uplink data correlation between user $i$ and other SBS $j$'s associated users. Similarly, finding the maximum data correlation allows minimizing the uplink transmission data that SBS $j$ uses to determine user $i$'s location and orientation. 

In the uplink, the tracking information can be directly processed by the SBSs that have limited computational power. The computational resource of each SBS, $c$,  represents its ability to compute the tracking data. Each SBS $j$ will allocate the total computational power to the associated users and, hence, $m_{ij}$ is used to represent the computational power that SBS $j$ allocates to user $i$ with $\sum\nolimits_{i \in {\mathcal{U}_j}} {{m_{ij}}}  = m$. $\mathcal{U}_j$ represents the set of the users associated with SBS $j$. The computation time of SBS $j$ that processes the tracking data collected by user $i$ is $\frac{K_{i}\left(\sigma_{i}^{\max}\right)}{m_{ij}}$ and the total uplink delay can be given by:
\begin{equation}
\setlength{\abovedisplayskip}{3 pt}
\setlength{\belowdisplayskip}{3 pt}
D_{ij}^\textrm{u}\left(K_{i}\!\left(\sigma_{i}^{\max}\right)\!,\boldsymbol{v}_{ij},m_{ij}\right) = \frac{{{K_{i}\!\left(\sigma_{i}^{\max}\right)}}}{{{c_{ij}\left(\boldsymbol{v}_{ij}\right)}}}+\frac{K_{i}\!\left(\sigma_{i}^{\max}\right)}{m_{ij}},
\end{equation}
where the first term is the transmission time from user $i$ to SBS $j$ and the second term is the computation time for user $i$' data. The computation time depends on the computational resources that SBS $j$ allocates to each user that will affect the uplink delay.      

 \subsection{Utility Function Model}
 In order to jointly consider the transmission delay in both uplink and downlink, we introduce a method based on the framework of multi-attribute utility theory \cite{abbas2010constructing} to construct an appropriate utility function to capture transmission delay in both uplink and downlink. We first introduce the utility functions of transmission delay in uplink and downlink, separately. Then, we formulate the utility function based on \cite{abbas2010constructing}.
 
 The utility function of downlink transmission delay is constructed based on the normalization of downlink transmission delay, which can be given by:
 \begin{equation}\label{eq:delayd}
 \setlength{\abovedisplayskip}{4 pt}
\setlength{\belowdisplayskip}{4 pt}
 \begin{split}
  &\bar D_{ij}\left(L_{i}\left(\phi _{i}^{\max} \right),\boldsymbol{s}_{ij}\right)=\\&\!\!\!\left\{ {\begin{array}{*{20}{c}}
{\!\!\frac{{D_{ij,\max}-{D_{ij}}\left( {{L_{i}}\left( {\phi _{i}^{\max }} \right),{\boldsymbol{s}_{ij}}} \right)}}{{D_{ij,\max}- {\gamma _D}}},{D_{ij}}\left( {{L_{i}}\left( {\phi _{i}^{\max }} \right)\!,{\boldsymbol{s}_{ij}}} \right) \ge {\gamma _D},}\\
{\;\;\;\;\;\;\;\;\;\;\;\;\;\;\;\;1,\;\;\;\;\;\;\;\;\;\;\;\;\;\;\;\;{D_{ij}}\left( {{L_{i}}\left( {\phi _{i}^{\max }} \right)\!,{\boldsymbol{s}_{ij}}} \right) < {\gamma _D},}
\end{array}} \right.
 \end{split}
 \end{equation}
 where $\gamma_{D}$ is the maximal tolerable delay for each VR user (maximum supported by the VR system being used) and $D_{ij,\max}=\mathop {\max }\limits_{{\boldsymbol{s}_{ij}}} \left( {{D_{ij}}\left( {{L_{i}}\left( 0 \right),{\boldsymbol{s}_{ij}}} \right)} \right)$ is the maximum transmission delay. From (\ref{eq:delayd}), we can see that, when the downlink transmission delay is smaller than $\gamma_D$, the utility value will remain at 1. This is due to the fact when the delay meets the system requirement, the network will encourage the SBSs to reallocate the resource blocks to other users. 
 The utility function for the uplink transmission is:
  \begin{equation}\small\label{eq:delayu}
  \setlength{\abovedisplayskip}{4 pt}
\setlength{\belowdisplayskip}{4 pt}
 \begin{split}
  &\bar D_{ij}^\textrm{u}\left(K_{i}\left(\sigma _{i}^{\max} \right),\boldsymbol{v}_{ij},m_{ij}\right)=\\&\!\!\!\!\left\{ {\begin{array}{*{20}{c}}
{\!\!\!\!\!\frac{{D_{ij,\max}^\textrm{u}\!-{D_{ij}^\textrm{u}}\left( {{K_{i}}\left( {\sigma _{i}^{\max }} \right),{\boldsymbol{v}_{ij}},m_{ij}} \right)}}{{D_{ij,\max}^\textrm{u}- {\gamma _D}}},{D_{ij}^\textrm{u}}\!\left( {{K_{i}}\!\left( {\sigma _{i}^{\max }} \right)\!,{\boldsymbol{v}_{ij}},m_{ij}} \right)\! \ge \!{\gamma _D^\textrm{u}},}\\
{\;\;\;\;\;\;\;\;\;\;\;\;\;\;\;\;\;\;\;\;1,\;\;\;\;\;\;\;\;\;\;\;\;\;\;\;\;\;\;\;\;\;\!{D_{ij}^\textrm{u}}\!\left( {{K_{i}}\!\left( {\sigma _{i}^{\max }} \right)\!,{\boldsymbol{v}_{ij}},m_{ij}} \right) \!<\! {\gamma _D^\textrm{u}},}
\end{array}} \right.
 \end{split}
 \end{equation}
 where $\gamma_{D}^\textrm{u}$ is the maximal tolerable delay for the VR tracking information transmission and $D_{ij,\max}^\textrm{u}=\mathop {\max }\limits_{{\boldsymbol{v}_{ij},m_{ij}}} \left( {{D_{ij}^\textrm{u}}\left( {{K_{i}}\left( 0 \right),{\boldsymbol{v}_{ij}},m_{ij}} \right)} \right)$ is the maximal uplink delay. Based on (\ref{eq:delayd}) and (\ref{eq:delayu}), the total utility function that captures both downlink and uplink delay for user $i$ associated with SBS $j$ is:
 \begin{equation}
 \setlength{\abovedisplayskip}{4 pt}
\setlength{\belowdisplayskip}{4 pt}
 \begin{split}
&\!\!\!\!\!\!\!\!\!\!{U_{ij}}\left( {{\boldsymbol{s}_{ij}},{\boldsymbol{v}_{ij}},{{m}_{ij}}} \right) =\\& {\bar D_{ij}}\left( {{L_{i}}\left( {\phi _{i}^{\max }} \right),{\boldsymbol{s}_{ij}}} \right)\bar D_{ij}^\textrm{u}\left( {{K_{i}}\left( {\sigma _{i}^{\max }} \right),{\boldsymbol{v}_{ij}},{m_{ij}}} \right).
 \end{split}
 \end{equation}
 Here, ${L_{i}}\left( {\phi _{i}^{\max }} \right)$ and ${K_{i}}\left( {\sigma _{i}^{\max }} \right)$ are determined by the user association scheme. To capture the gain that stems from the allocation of the resource blocks and the computational capabilities, we state the following result:
 
 \begin{theorem}\label{th:1}
\emph{The utility gain of user $i$'s delay due to an increase in the amount of allocated resource blocks and computational resources is:}

\emph{\romannumeral1) The gain that stems from an increase in the allocated uplink resource blocks, $\Delta U_{ij}$, is given by:
\begin{equation}\small
\Delta U_{ij}=\left\{ {\begin{array}{*{20}{c}}
\;{{f_{\bar D_{ij}^\textrm{u}}}\left( {\frac{1}{{{c_{ij}}\left( {{\boldsymbol{v}_{ij}}} \right)}}} \right),\;\;\;\;\;{c_{ij}}\left( {\Delta {\boldsymbol{v}_{ij}}} \right) \gg {c_{ij}}\left( {{\boldsymbol{v}_{ij}}} \right),}\\
\;{{f_{\bar D_{ij}^\textrm{u}}}\left( {\frac{{{c_{ij}}\left( {\Delta {\boldsymbol{v}_{ij}}} \right)}}{{{c_{ij}}{{\left( {{\boldsymbol{v}_{ij}}} \right)}^2}}}} \right),\;\;\;{c_{ij}}\left( {\Delta {\boldsymbol{v}_{ij}}} \right) \ll {c_{ij}}\left( {{\boldsymbol{v}_{ij}}} \right),}\\
{{f_{\bar D_{ij}^\textrm{u}}}\left( {\frac{{{c_{ij}}\left( {\Delta {\boldsymbol{v}_{ij}}} \right)}}{{{c_{ij}}{{\left( {{\boldsymbol{v}_{ij}}} \right)}^2} + {c_{ij}}\left( {{\boldsymbol{v}_{ij}}} \right){c_{ij}}\left( {\Delta {\boldsymbol{v}_{ij}}} \right)}}} \right),\;\;\;\textrm{else},}
\end{array}} \right.
\end{equation}
where ${{f_{\bar D_{ij}^\textrm{u}}}\left( x \right)}={\bar D_{ij}}\left( {{L_{i}}\left( {\phi _{i}^{\max }} \right),{\boldsymbol{s}_{ij}}} \right)\left(\frac{{{K_{i}\left(\sigma_{i}^{\max}\right)}x }}{{D_{ij,\max}^\textrm{u}- {\gamma _D^\textrm{u}}}}\right)$.}
 
\emph{\romannumeral2) The gain that stems from the increase in the number of downlink resource blocks allocated to user $i$, $\Delta U_{ij}$, is:
\begin{equation}
\Delta {U_{ij}} ={\left\{ {\begin{array}{*{20}{c}}
{{f_{{{\bar D}_{ij}}}}\left( {\frac{1}{{{c_{ij}}\left( {{\boldsymbol{s}_{ij}}} \right)}}} \right),\;\;{c_{ij}}\left( {\Delta {\boldsymbol{s}_{ij}}} \right) \gg {c_{ij}}\left( {{\boldsymbol{s}_{ij}}} \right),}\\
{{f_{{{\bar D}_{ij}}}}\left( {\frac{{{c_{ij}}\left( {\Delta {\boldsymbol{s}_{ij}}} \right)}}{{{c_{ij}}{{\left( {{\boldsymbol{s}_{ij}}} \right)}^2}}}} \right),{c_{ij}}\left( {\Delta {\boldsymbol{s}_{ij}}} \right) \ll {c_{ij}}\left( {{\boldsymbol{s}_{ij}}} \right),}\\
{{f_{{{\bar D}_{ij}}}}\left( {\frac{{{c_{ij}}\left( {\Delta {\boldsymbol{s}_{ij}}} \right)}}{{{c_{ij}}{{\left( {{\boldsymbol{s}_{ij}}} \right)}^2} + {c_{ij}}\left( {{\boldsymbol{s}_{ij}}} \right){c_{ij}}\left( {\Delta {\boldsymbol{s}_{ij}}} \right)}}} \right),\;\;\;\;\textrm{else},}
\end{array}} \right.}
\end{equation}
where ${{f_{\bar D_{ij}}}\!\!\left( x \right)}\!=\!{\bar D_{ij}^\textrm{u}}\!\left( {{K_{i}}\left( {\sigma _{i}^{\max }} \right)\!,{\boldsymbol{v}_{ij}},m_{ij}} \right)\!\left(\!\!\frac{{{L_{i}\left(\phi_{i}^{\max}\right)}x}}{{D_{ij,\max}- {\gamma _D}}}\!\right)$.}

\emph{\romannumeral3) The gain that stems from the increase in the amount of computational resources, $\Delta m$, allocated to user $i$, $\Delta U_{ij}$, is:
\begin{equation}\small
\Delta {U_{ij}} ={\bar D_{ij}}\!\left( {{L_{i}}\left( {\phi _{i}^{\max }} \right),{\boldsymbol{s}_{ij}}} \right)\!\!\left(\!\frac{{K_{i}\!\left(\sigma_{i}^{\max}\right)\!\Delta m}}{{\left(\!D_{ij,\max}^\textrm{u}\!\!- \!{\gamma _D^\textrm{u}}\right)\!\!\left({m_{ij}}\!\!\left( {{m_{ij}} \!\!+ \!\!\Delta m} \right)\right)}}\!\right)\!.
\end{equation}}
\end{theorem}
\begin{proof} The details of the proof are found in \cite{proof} 
\end{proof}
From Theorem \ref{th:1}, we can see that the allocation of spectrum and computational resources jointly determines the delay utility. Indeed, Theorem \ref{th:1} will provide guidance for the SBSs when they select actions in the learning algorithm proposed in Section \ref{se:al}.     
 \subsection{Problem Formulation}
 Given the defined system model, our goal is to develop an effective resource allocation scheme that allocates resource blocks and computational power to maximize the utility functions of all users. However, the maximization problem depends not only on the allocation of resource blocks and computational resources but also on the user associations. Moreover, the utility value of each SBS depends not only on its own choice of resource allocation scheme but also on the remaining SBSs' schemes. In addition, the data correlation among the users varies as the period changes, which will affect the resource allocation and user association. In this case, we first formulate a noncooperative game $\mathcal{G} = \left[ {\mathcal{R},\left\{ {{\mathcal{A}_j}} \right\}_{j \in \mathcal{R}},\left\{ {{U_j}} \right\}}_{j \in \mathcal{R}} \right]$. In this game, the players are the SBSs, $\mathcal{A}_j$ represents the action set of each SBS $j$, and $U_j$ is the utility function of each SBS $j$. The action $\boldsymbol{a}_j$ of SBS $j$ consists of: (i) downlink resource allocation vector $\boldsymbol{s}_j=\left[\boldsymbol{s}_{1j},\boldsymbol{s}_{2j}, \dots , \boldsymbol{s}_{\mathcal{U}_jj}\right]$, (ii) uplink resource allocation vector $\boldsymbol{v}_j=\left[\boldsymbol{v}_{1j},\boldsymbol{v}_{2j}, \dots , \boldsymbol{v}_{\mathcal{U}_jj}\right]$, and (iii) computational resource allocation vector $\boldsymbol{m}_j=\left[m_{1j}, m_{2j}, \dots , m_{\mathcal{U}_jj}\right]$. Here, $m_{ij} \in \mathcal{M}, i \in \mathcal{U}_j$ where $\mathcal{M}= \left\{\frac{c}{M},\frac{2c}{M},...,c\right\}$ is a finite set of M level fractions of SBS $j$'s total computational resource $m_j$. We assume that each SBS $j$ adopts one action at each time slot $t$. Then, the utility function of each SBS $j$ can be given by:          
 \begin{equation}\label{eq:u}
\setlength{\abovedisplayskip}{4 pt}
\setlength{\belowdisplayskip}{4 pt}
u_{j} \left(\boldsymbol{a}_{j}, {\boldsymbol{a}_{-j}}\right)=\!\!\mathop \frac{1}{T}\sum\limits_{t = 1}^T \sum\limits_{i \in \mathcal{U}_j}{ U_{ij,t}\left( { {{\boldsymbol{s}_{ij},\boldsymbol{v}_{ij}}}, \boldsymbol{c}_{ij}} \right)},
\end{equation}                                                                
where $\boldsymbol{a}_{j} \in {\mathcal{A}_{j}}$ is an action of SBS $j$ and $\boldsymbol{a}_{-j}$ denotes the action profile of all SBSs other than SBS $j$. Indeed, (\ref{eq:u}) captures the average utility value of each SBS $j$. Let ${\pi _{j,\boldsymbol{a}_{ij}}} \!=\! \!\frac{1}{T}\!\!\sum\limits_{t = 1}^T\!\mathbbm{1}_{\left\{ \boldsymbol{a}_{j,t}=\boldsymbol{a}_{ij} \right\}}\!\!= \Pr \left( {{\boldsymbol{a}_{j,t}} = {\boldsymbol{a}_{ij}}} \right)$ be the probability of SBS $j$ using action $\boldsymbol{a}_{ij}$. $\boldsymbol{a}_{j,t}$ represents the action that SBS $j$ uses at time $t$ and ${{\boldsymbol{a}_{j,t}} = {\boldsymbol{a}_{ij}}}$ denotes that SBS $j$ adopts action $\boldsymbol{a}_{ij}$ at time $t$. ${\boldsymbol{\pi} _j}= \left[ {{\pi _{j,\boldsymbol{a}_{1j}}}, \ldots ,{\pi _{j,\boldsymbol{a}_{\left| {{\mathcal{A}_j}} \right|j}}}} \right]$ is the action selection mixed strategy of SBS $j$ with $\left| {{\mathcal{A}_j}} \right|$ being the number of actions of SBS $j$. Based on the definition of the strategy, the utility function in (\ref{eq:u}) is given by:
\begin{equation}\small
\setlength{\abovedisplayskip}{4 pt}
\setlength{\belowdisplayskip}{4 pt}
\begin{split}
u_{j} \!\left(\boldsymbol{a}_{j}, {\boldsymbol{a}_{-j}}\right)\!&\!=\!\! \frac{1}{T}\!\sum\limits_{t = 1}^T \!{ U_{j,t}\!\left( { {{\boldsymbol{a}_{j},\boldsymbol{a}_{-j}}}} \right)}\!=\!\!\sum\limits_{\boldsymbol{a} \in \mathcal{A}}\!\!\left(\! {{\!U_{j}}\!\left( {{\boldsymbol{a}_j},{\boldsymbol{a}_{ - j}}} \right)\!}\! \prod\limits_{j \in \mathcal{B}} {\!\pi _{j,\boldsymbol{a}_j}}\!\!\right),
\end{split}
\end{equation}
where $\boldsymbol{a} \in \mathcal{A}$ with $\mathcal{A}$ being the action set of all SBSs.   

One suitable solution for this game is the mixed-strategy Nash equilibrium (NE), formally defined as follows \cite{bacci2016game}:
 A mixed strategy profile ${\boldsymbol{\pi} ^*} = \left( {\boldsymbol{\pi} _1^*, \ldots ,\boldsymbol{\pi} _{B}^*} \right) = \left( {\boldsymbol{\pi} _j^*,\boldsymbol{\pi}_{ - j}^*} \right)$ is a \emph{mixed-strategy Nash equilibrium} if, $\forall j \in \mathcal{R}$ and $\boldsymbol{\pi} _j$, we have:
\begin{equation}\label{eq:mNE}
{ u_j}\left( {\boldsymbol{\pi} _j^*,\boldsymbol{\pi} _{ - j}^*} \right) \ge { u_j}\left( {{\boldsymbol{\pi} _j},\boldsymbol{\pi} _{ - j}^*} \right),
\end{equation}
where 
${ u_j}\left( {{\boldsymbol{\pi} _n},{\boldsymbol{\pi} _{ - n}}} \right) =\sum\limits_{{\boldsymbol{a}} \in {\mathcal{A}}} {{U_j}\left( {{\boldsymbol{a}_j,\boldsymbol{a}_{-j}}} \right)\prod\limits_{j \in \mathcal{B}} {{\pi _{j,{\boldsymbol{a}_j}}}}}$ is the expected utility of SBS $j$ when it selects the mixed strategy $\boldsymbol{\pi} _j$.
For our game, the mixed-strategy NE for the SBSs represents a solution of the game at which each SBS $j$ can minimize the delay for its associated users, given the actions of its opponents. 

 \section{Echo State Networks for Self-Organizing Resource Allocation} \label{se:al}
Next, we introduce a transfer reinforcement learning (RL) algorithm that can be used to find an NE of the VR game. To satisfy the delay requirement for the VR transmission, we propose a transfer RL algorithm based on the neural networks framework of ESNs [16]. Traditional RL algorithms such as Q-learning typically rely on a Q-table to record the utility value. However, as the number of players and actions increases, the number of utility values that the Q-table needs to include will increase exponentially and, {hence, the Q-table may not be able to record all of the needed utility values. However, the proposed algorithm uses a utility function approximation method to record the utility value and, hence, it can be used for large networks and large utility spaces.} Moreover, a dynamic network in which the users' computational resources and data correlation may change across the time, traditional RL algorithms need to be executed each time the network changes. However, the proposed ESN transfer RL algorithm can find the relationship of the utility functions when the environment changes. After learning this relationship, the proposed algorithm can use the historic learning result to find a mixed strategy NE. 

The proposed transfer RL algorithm consists of two components: (i) ESN-based RL algorithm and (ii) ESN-based transfer learning algorithm. The ESN-based RL algorithm is based on our work in \cite{VROWNchen}, and, thus, here, we just introduce the ESN-based transfer learning algorithm. 

We first assume that, before the users' state information changes, the mixed strategy, action, and utility of each SBS $j$ are $\boldsymbol{\pi}_j$, $\boldsymbol{a}_j$ and $\hat u_{j}\left(\boldsymbol{a}_j,\boldsymbol{a}_{-j}\right)$, while the strategy, action, and utility of SBS $j$, after the users' state information changes, are $\boldsymbol{\pi}_j$, $\boldsymbol{a}_j$, and $\hat u'_{j}\left(\boldsymbol{a}_j,\boldsymbol{a}_{-j}\right)$. Since the number of users associated with SBS $j$ is unchanged, the action and mixed strategy sets of SBS $j$ will not change when the users' state information changes.  In this case, the proposed ESN-based transfer learning algorithm is used to find the relationship between $\hat u_{j}\left(\boldsymbol{a}_j,\boldsymbol{a}_{-j}\right)$ and $\hat u'_{j}\left(\boldsymbol{a}_j,\boldsymbol{a}_{-j}\right)$ when SBS $j$ only knows $\hat u_{j}\left(\boldsymbol{a}_j,\boldsymbol{a}_{-j}\right)$. This means that the proposed algorithm can transfer the information from the already learned utility $\hat u_{j}\left(\boldsymbol{a}_j,\boldsymbol{a}_{-j}\right)$ to the new utility $\hat u'_{j}\left(\boldsymbol{a}_j,\boldsymbol{a}_{-j}\right)$ that must be learned. The ESN-based transfer learning algorithm of each SBS $j$ consists of three components: (a) input, (b) output, and (c) ESN model, which are given by: 

$\bullet$ \emph{Input:} The ESN-based transfer learning algorithm takes the strategies of the SBSs and the action of SBS $j$ uses at time $t$ as input which is given by $\boldsymbol{x}'_{t,j}=\left[ {\pi_{1}, \cdots , \pi_{B}}, \boldsymbol{a}_{j,t} \right]^{\mathrm{T}}$.

$\bullet$ \emph{Output:} The output of the ESN-based transfer learning algorithm at time $t$ is the deviation of the utility values when the users' information changes ${y}'_{j,t}=\hat u'_{j}\left( \boldsymbol{a}_{j,t}\right)-\hat u_{j}\left( \boldsymbol{a}_{j,t}\right)$.

$\bullet$ \emph{ESN Model:} An ESN model is used to find the relationship between the input $\boldsymbol{x}'_{t,j}$ and output ${y}'_{t,j}$. The ESN model consists of the output weight matrix $\boldsymbol{W}_j^{'\textrm{out}} \in {\mathbb{R}^{1 \times N_w}}$ and the dynamic reservoir containing the input weight matrix $\boldsymbol{W}_j^{'\textrm{in}} \in {\mathbb{R}^{N_w \times B+1}}$, and the recurrent matrix $\boldsymbol{W}'_j \in \mathbb{R}^{N_w \times N_w}$ with $N_w$ being the number of the dynamic reservoir units. Here, the dynamic reservoir is used to store historic ESN information that includes input, reservoir state, and output. This information is used to build the relationship between the input and output. The update process of the dynamic reservoir will be given by: 
 \begin{equation}\label{eq:state2}
 \setlength{\abovedisplayskip}{4 pt}
\setlength{\belowdisplayskip}{4 pt}
{\boldsymbol{\mu}'_{j,t}} ={\mathop{f}\nolimits}\!\left( {\boldsymbol{W}'_j{\boldsymbol{\mu}'_{j,t-1}} + \boldsymbol{W}_j^{'\textrm{in}}{\boldsymbol{x}'_{j,t}}} \right).
\end{equation}
where  $f\!\left(x\right)=\frac{{{e^x} - {e^{ - x}}}}{{{e^x} + {e^{ - x}}}}$ is the tanh function. Based on the dynamic reservoir state, the ESN-based transfer learning algorithm will combine with the output weight matrix to approximate the deviation of the utility value, which can be given by:
\begin{equation}\label{eq:update}
\setlength{\abovedisplayskip}{4 pt}
\setlength{\belowdisplayskip}{4 pt}
{y}'_{j,t} = {\boldsymbol{W}_{j,t}^{'\textrm{out}}} {{\boldsymbol{\mu}'_{j,t}}},
\end{equation}
where ${\boldsymbol{W}_{j,t}^{'\textrm{out}}}$ is the output weight matrix at time slot $t$. 
\begin{equation}\label{eq:w2}
\setlength{\abovedisplayskip}{4 pt}
\setlength{\belowdisplayskip}{4 pt}
{\boldsymbol{W}_{j,t + 1}^{'\textrm{out}}} = {\boldsymbol{W}_{j,t}^{'\textrm{out}}} + {\lambda'} \left( {\hat u'_{j}\left( \boldsymbol{a}_{j,t}\right)-\hat u_{j}\left( \boldsymbol{a}_{j,t}\right) -y'_{j,t}} \right){{\boldsymbol{\mu}'}_{j,t}^{\mathrm{T}}},
\end{equation}
where $\lambda'$ is the learning rate, and $\hat u'_{j,t}$ is the actual deviation between two utility values. In this case, the ESN-based transfer learning algorithm can find the relationship between the utility functions when the users' state information changes and, {hence, reduce the iterations of the RL algorithm to learn the new utility values.} The proposed, distributed ESN-based learning algorithm performed by each SBS $j$ is summarized in Table~I. {The proposed algorithm is guaranteed to converge to an NE and this convergence follows from \cite{VROWNchen}.}


\begin{table}[!t]\label{tb1}
  \centering
  \caption{
    \vspace*{-0.2em} ESN-based Learning Algorithm for resource Allocation}\vspace*{-1em}
    \begin{tabular}{p{3.5in}}
      \hline \vspace*{-0.8em}
      \textbf{Inputs:}\,\, $\boldsymbol{x}_{j,t}$ and $\boldsymbol{x}'_{j,t}$  \vspace*{-0.5em}\\
\hspace*{1em}\textit{Initialize:}   \vspace*{-0.3em}
$\boldsymbol{W}_j^{\textrm{in}}$, $\boldsymbol{W}_j$, $\boldsymbol{W}_j^{\textrm{out}}$, $\boldsymbol{W}_j^{'\textrm{in}}$, $\boldsymbol{W}'_j$, $\boldsymbol{W}_j^{'\textrm{out}}$, $\boldsymbol{y}_{j}=0$, and ${y}'_{j}=0$.
\vspace*{-0.1em}
\hspace*{0em}\begin{itemize}\vspace*{-0.1em}
\item[] \hspace*{0em} \textbf{for} each time $t$ \textbf{do}.
\item[] \hspace*{0.8em}(a) Estimate the value of the utility function $\hat u_{j,t}$ based on (\ref{eq:update}).
\item[] \hspace*{0.3em} \textbf{if} $t==1$ 
\item[] \hspace*{0.8em}(b) Set the mixed strategy $\boldsymbol{\pi}_{j,t}$ uniformly.
\item[] \hspace*{0.3em} \textbf{else}
\item[] \hspace*{1.8em}(c) Set the mixed strategy $\boldsymbol{\pi}_{j,t}$ based on the $\varepsilon$-greedy exploration.
\item[] \hspace*{0.3em} \textbf{end if}
\item[] \hspace*{0.8em}(d) Broadcast the index of the mixed strategy to other SBSs.
\item[] \hspace*{0.8em}(e) Receive the index of the mixed strategy as input $\boldsymbol{x}_{j,t}$.
\item[] \hspace*{0.8em}(f) Perform an action based on the mixed strategy. 
\item[] \hspace*{0.8em}(g) Use the index of the mixed strategies and action as input $\boldsymbol{x}'_{j,t}$.
\item[] \hspace*{0.8em}(h) Estimate the value of the difference of utility function $y'_{j,t}$.
\item[] \hspace*{0.8em}(i) Update the dynamic reservoir state $\boldsymbol{\mu}_{j,t}$.
\item[] \hspace*{0.8em}(j) Update the output weight matrix $\boldsymbol{W}_j^{\textrm{out}}$ based on $y'_{j,t}$.
\item[] \hspace*{0em} \textbf{end for}
\end{itemize}\vspace*{-0cm}\vspace*{-0.8em}\\
   \hline
    \end{tabular}\label{tab:algo}\vspace{-0.4cm}
\end{table} 

\section{Simulation Results}
\vspace{-0.1cm}
For our simulations, we consider an SCN deployed within a circular area with radius $r = 100$ m. $U=25$ users and $B=4$ SBSs are uniformly distributed in this SCN area. The rate requirement of VR transmission is 25.32 Mbit/s \cite{VROWNchen}. 
The detailed parameters are listed in Table  \uppercase\expandafter{\romannumeral3}. For comparison purposes, we use ESN algorithm and a baseline Q-learning algorithm in \cite{VROWNchen}. 

\begin{table}
  \newcommand{\tabincell}[2]{\begin{tabular}{@{}#1@{}}#2\end{tabular}}
\renewcommand\arraystretch{0.8}
 \caption{
    \vspace*{-0.2em}\tiny SYSTEM PARAMETERS}\vspace*{-1em}
\centering  
\begin{tabular}{|c|c|c|c|}
\hline
\textbf{Parameter} & \textbf{Value} & \textbf{Parameter} & \textbf{Value} \\
\hline
$F$ & 1000 & $P_B$ & 20 dBm\\
\hline
$B$ & 2 MHz & $S$, $V$ & 5, 5\\
\hline
$N_w$ & 1000 & $\sigma ^2$ & -95 dBm\\
\hline
$ N_{v}$ &6&$\lambda$, $\lambda'$ & 0.03, 0.3 \\
\hline
$m$&5&$r_B$& $30$ m\\
\hline
$\alpha$ & 2 & $V_{F}$ & 100 Gbit/s\\
\hline
\end{tabular}
 \vspace{-0.5cm}
\end{table}  

 \begin{figure}[!t]
  \begin{center}
   \vspace{0cm}
    \includegraphics[width=6.5cm]{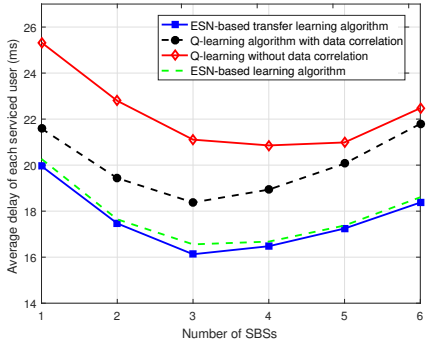}
    \vspace{-0.3cm}
    \caption{\label{figure4} Average delay of each user vs. number of SBSs.}
  \end{center}\vspace{-1cm}
\end{figure}
\vspace{-0cm}
Fig. \ref{figure4} shows how the average delay per user changes with the number of SBSs. Fig. \ref{figure4} shows that, as the number of SBSs increases, the average delay of all algorithms decreases, then increases. This is due to the fact that, as the number of SBSs increases, the number of users located in each SBS's coverage decreases and, hence, the average delay decreases. However, as the number of SBSs keeps increasing, the interference will also increase.
Fig. \ref{figure4} also shows that our algorithm achieves up to 16.7\% and 18.2\% gains in terms of average delay compared to the Q-learning with data correlation and Q-learning without data correlation for 6 SBSs.   
This is due to the fact that our algorithm can transfer information across time. From Fig. \ref{figure4}, we can also see that the deviation between Q-learning algorithms decreases as the number of SBSs changes. In fact, as the number of SBSs increases, the number of users associated with each SBS decreases and, hence, the data correlation of users decreases. {Fig. \ref{figure4} also shows that the delay gain of the proposed algorithm is small compared with ESN algorithm. However, the proposed algorithm can converge much faster as shown in Fig. \ref{figure6}.}

 \begin{figure}[!t]
  \begin{center}
   \vspace{0cm}
    \includegraphics[width=6.5cm]{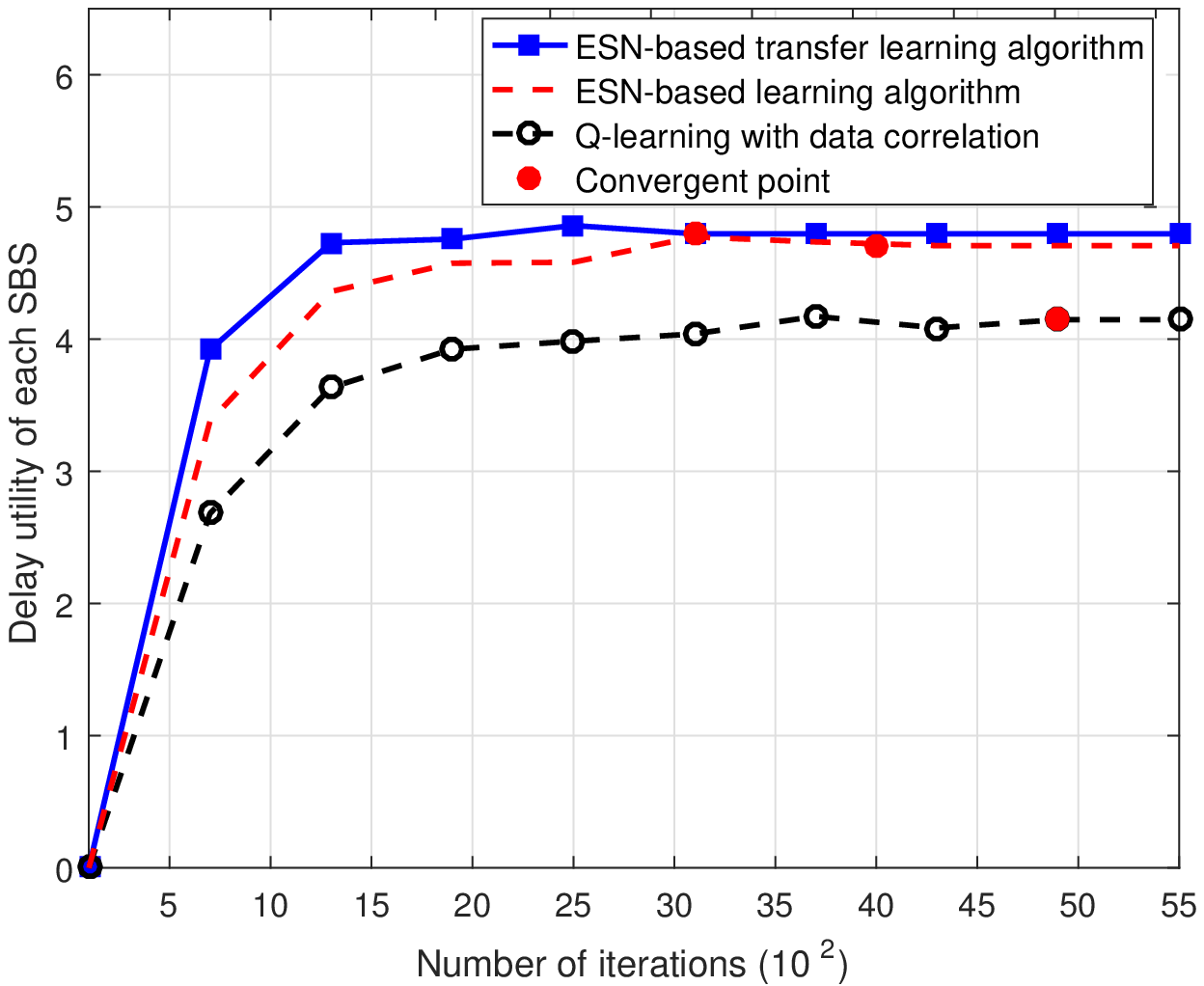}
    \vspace{-0.3cm}
    \caption{\label{figure6}Convergence of the proposed algorithm and Q-learning.}
  \end{center}\vspace{-0.9cm}
\end{figure}
\vspace{-0cm}

Fig. \ref{figure6} shows the number of iterations needed till convergence for the proposed approach, ESN algorithm, and Q-learning with data correlation when the users' information changes. In this figure, we can see that, as time elapses, the delay utilities for all considered algorithms increase until convergence to their final values. Fig. \ref{figure6} also shows that the proposed algorithm achieves, respectively, 22.5\% and 36\% gains in terms of the number of the iterations needed to reach convergence compared to ESN algorithm and Q-learning. This implies that the proposed algorithm can apply the already learned utility value to the new utility value that must be learned as the users' information changes.


\section{Conclusion}
In this paper, we have proposed a novel resource allocation framework for optimizing delay for wireless VR services with data correlation. We have formulated the problem as a noncooperative game and proposed a novel transfer learning algorithm based on echo state networks to solve the game. 
The proposed learning algorithm can use the existing learning result to directly find the optimal resource allocation when the users' state information changes and, hence, can quickly converge to a mixed-strategy NE. Simulation results have shown that the proposed algorithm has a faster convergence time than Q-learning and guarantees low delays for VR services. 

\vspace{-0.1cm}
\bibliographystyle{IEEEbib}
\def\baselinestretch{0.7}
\bibliography{references}
\end{document}